\documentclass[12pt]{article}
\usepackage{amsfonts}
\usepackage{amsmath}
\usepackage{amssymb}
\usepackage{amsthm} 
\usepackage{epsfig}
\usepackage{tikz}
\usetikzlibrary{matrix,arrows,decorations.pathmorphing}

\usepackage{chicago}
\usepackage{psfrag}
\usepackage{geometry}
\usepackage{color}
\usepackage{rotating}
\usepackage{bm}

\newcommand{\e}{{\rm e}}

\newcommand{\E}{{\mathbb E}}
\newcommand{\Pa}{{\mathbb P}}
\newcommand{\Q}{{\mathbb Q}}

\newcommand{\R}{{\mathbb R}}
\newcommand{\La}{{\mathbb L}}

\newcommand{\Acal}{{\mathcal A}}
\newcommand{\Bcal}{{\mathcal B}}
\newcommand{\Ccal}{{\mathcal C}}
\newcommand{\Dcal}{{\mathcal D}}
\newcommand{\Ecal}{{\mathcal E}}
\newcommand{\Fcal}{{\mathcal F}}

\newcommand{\Ucal}{{\mathcal U}}
\newcommand{\Mcal}{{\mathcal M}}

\newtheorem{proposition}{Proposition}[section]
\newtheorem{lemma}[proposition]{Lemma}
\newtheorem{theorem}[proposition]{Theorem}
\newtheorem{definition}[proposition]{Definition}
\newtheorem{corollary}[proposition]{Corollary}
\newtheorem{remark}[proposition]{Remark}
\newtheorem{exampleemph}[proposition]{Example}   

\newenvironment{example}{\begin{exampleemph}\begin{upshape}}{\end{upshape}\end{exampleemph}} 
\newtheorem{foo}[proposition]{Remarks}

\begin{document}

\title{On the Relation Between Linearity-Generating Processes and Linear-Rational Models\footnote{
The authors wish to thank participants at the ETH-ITS Workshop on Mathematical Finance Beyond Classical Models, LSE Risk and Stochastics Conference 2018, Aarhus Econometrics-Finance Seminar, SFI Research Days Gerzensee 2018, Paul Schneider, J\'er\^ome Detemple and Vadim Linetsky (the editors), and two anonymous referees for comments.
The research leading to these results has received funding from the European Research Council under the European Union's Seventh Framework Programme (FP/2007-2013) / ERC Grant Agreement n.~307465-POLYTE.}
}

\author{
Damir Filipovi\'c\thanks{EPFL and Swiss Finance Institute. Email: damir.filipovic@epfl.ch} \quad\quad
Martin Larsson\thanks{ETH Zurich. Email: martin.larsson@math.ethz.ch} \quad\quad
Anders B. Trolle\thanks{HEC Paris. Email: trolle@hec.fr}}
\date{June 5, 2018}
\maketitle

\begin{abstract}
We review the notion of a linearity-generating (LG) process introduced by \citeN{Gabaix:2007} and relate LG processes to linear-rational (LR) models studied by \citeN{FilipovicLarssonTrolle14}. We show that every LR model can be represented as an LG process and vice versa. We find that LR models have two basic properties which make them an important representation of LG processes. First, LR models can be easily specified and made consistent with nonnegative interest rates. Second, LR models go naturally with the long-term risk factorization due to \citeN{AlvarezJermann2005}, \citeN{HansenScheinkman2009}, and  \citeN{qin_lin_17}. Every LG process under the long forward measure can be represented as a lower dimensional LR model.
\end{abstract}

\section{Introduction}

Linearity-generating (LG) processes constitute an important class of stochastic models in finance that yield linear asset prices. LG processes were introduced by \citeN{Gabaix:2007} in discrete time and in continuous time. LG processes were further studied by \citeN{CheriditoGabaix2008} and subsequently applied by \citeN{CarrGabaixWu2009} for modeling interest rates, and by \citeN{Gabaix2012} and \citeN{Gabaix2016} for solving puzzles in macro-finance.

\citeN{FilipovicLarssonTrolle14} have recently introduced the class of linear-rational (LR) models for the term structure of interest rates in continuous time. The state price density $\zeta_t$ is a linear function in the factor process,  $\zeta_t = \e^{-\alpha t}(\phi +\psi^\top Z_t)$, which has a linear drift, $dZ_t = (b+\beta Z_t)\,dt + dM^Z_t$.  Nominal and deflated bond prices become linear-rational and linear in the factor process, respectively. LR models are related to LG processes, but the exact mapping is not straightforward.

In this paper, we rigorously relate LR models to LG processes in continuous time. We first review the definition of an LG process and provide a set of equivalent characterizations. We then show that every LG process $(\zeta_t, X_t)$ can naturally be represented as an LR model given by $Z_t=(\zeta_t;\zeta_t X_t)$. Conversely, we find that every $m$-dimensional LR model can be represented as an $(m+1)$-dimensional LG process. But this mapping is not surjective. Specifically, we call two LR models \emph{observationally equivalent} if they induce the same normalized state price densities and thus are indistinguishable from an economic point of view. Under some mild non-degeneracy conditions, we fully characterize the set of \emph{reducible} $(m+1)$-dimensional LG processes, namely those which are observationally equivalent to some $m$-dimensional LR models. Any such LR model in turn is shown to be \emph{proper}, in the sense that it cannot be represented as an LG process of the same or lower dimension.

We find that LR models have two basic properties which make them an important representation of LG processes. First, LR models can be easily specified and made consistent with nonnegative interest rates. Second, we show that the state price density specification in an LR model goes naturally with the long-term factorization of the state price density into a transitory and permanent component due to \citeN{AlvarezJermann2005} and \citeN{HansenScheinkman2009} in Markovian environments, and extended to a general semimartingale environment in \citeN{qin_lin_17}. Specifically, we prove that any proper LR model whose drift matrix $\beta$ has only eigenvalues with negative real parts yields the transitory component of the state price density. In other words, the reference probability measure of such an LR model coincides with the long forward measure, under which the gross return on the investment of one dollar at time zero in the zero-coupon bond of asymptotically long maturity is growth optimal.\footnote{ The long forward measure was first used by \citeN{FlesakerHughston1996} who called it the terminal measure.} The long-term yield in turn is shown to be constant and equal to $\alpha$. As main result we show that every LG process under the long forward measure is reducible, and hence observationally equivalent to a lower dimensional proper LR model.

The remainder of the paper is as follows.
Section~\ref{secLGchar} reviews LG processes.
Section~\ref{secspecLRNew} discusses the specification of LR models.
Section~\ref{secLRNew} contains our main results on the relation between LG processes and LR models.
Section~\ref{secLTF} shows that LR models appear naturally in the context of the long-term risk factorization due to \citeN{AlvarezJermann2005} and \citeN{HansenScheinkman2009}. Section~\ref{secconc} concludes.
All proofs are provided in the Appendix.

\section{Characterization of LG processes}\label{secLGchar}
We rigorously review the notion of an LG process. Throughout, we fix a filtered probability space $(\Omega, \Fcal_t,\Fcal,{\Pa^\ast})$. The measure ${\Pa^\ast}$ represents an auxiliary measure, which is locally equivalent but not necessarily identical to the objective measure $\Pa$. We write $\E^\ast_t[\cdot]$ for the $\Fcal_t$-conditional ${\Pa^\ast}$-expectation operator. Equalities between random variables are understood to hold almost surely. For more terminology and background of stochastic processes we refer to \citeN{jac_shi_03}.

Throughout this paper, we consider as state price density process an integrable positive semimartingale with multiplicative decomposition of the form
\begin{equation}\label{eqspddef}
  \zeta_t = \E^\ast[\zeta_0] \e^{-\int_0^t r_s \,ds} D_t
\end{equation}
for some short rate process $r_t$ and positive martingale $D_t$ with $\E^\ast[D_t]=1$. The corresponding risk-neutral measure $\Q$ is equivalent to $\Pa^\ast$ on each $\Fcal_t$ with Radon--Nikodym density $D_t$.\footnote{The risk-neutral measure $\Q$ is only locally defined on each $\Fcal_t$ but not globally on $\Fcal_\infty$, unless $D_t$ is uniformly integrable.} The state price density process $\zeta_t$ is a supermartingale if and only if the short rate process is nonnegative, $r_t\ge 0$. Let $X_t$ be an $n$-dimensional semimartingale such that $\zeta_t X_t$ is integrable.

We now give an integral-form version of the definition of a linearity-generating process introduced by~\citeN{Gabaix:2007}.\footnote{\citeN{Gabaix:2007} introduced LG processes in continuous time as well as discrete time. In this paper, we focus on continuous time.}

\begin{definition}\label{defLG}
The pair $(\zeta_t,X_t)$ forms an $(n+1)$-dimensional \emph{linearity-generating (LG) process} if there exists some continuously differentiable functions $\Acal$, $\Bcal$, $\Ccal$, $\Dcal$ with values in $\R$, $\R^{1\times n}$, $\R^{n\times 1}$, $\R^{n\times n}$, respectively, such that
\begin{equation}\label{eq1}
\begin{aligned}
\E^\ast_t\left[ \frac{\zeta_T}{\zeta_t}\right] &= \Acal(T-t)  + \Bcal(T-t)X_t \\
\E^\ast_t\left[ \frac{\zeta_T}{\zeta_t} X_T\right] &= \Ccal(T-t)  + \Dcal(T-t)X_t
\end{aligned}
\end{equation}
for all $0\le t\le T<\infty$.
\end{definition}

An LG process thus yields a linear time-$t$ price in $X_t$ of any contingent claim with linear time-$T$ payoff in $X_T$. For example, the first line in \eqref{eq1} is the price of a zero-coupon bond maturing at $T$,
\begin{equation}\label{eqPtT}
 P(t,T) =  \Acal(T-t)  + \Bcal(T-t)X_t .
\end{equation}

\begin{remark}
In \citeN{Gabaix:2007} the state price density $\zeta_t$ is replaced by the more general expression ``$M_t D_t$'', a pricing kernel ``$M_t$'' times a dividend ``$D_t$''. This allows for linear pricing of dividend paying assets. For zero-coupon bonds we have ``$D_t$''=1, and this is what we focus on in this paper.
\end{remark}

We next provide an equivalent characterization of the LG property. We define the \emph{affine support} and the \emph{linear support} of an $m$-dimensional semimartingale $Z_t$ as
\[ {\rm aff}(Z_\cdot) = \bigcap_{\substack{\text{$A$ affine subspace of $\R^m$} \\ \text{$Z_t\in A$ for all $t\ge 0$}}} A \]
and
\[ {\rm lin}(Z_\cdot) = \bigcap_{\substack{\text{$A$ linear subspace of $\R^m$} \\ \text{$Z_t\in A$ for all $t\ge 0$}}} A .\]
Note that ${\rm aff}(Z_\cdot)\subseteq {\rm lin}(Z_\cdot)$, with equality if and only if $0\in {\rm aff}(Z_\cdot)$. If ${\rm aff}(Z_\cdot)=\R^m$ then any affine function $h+H^\top z$ on $\R^m$ is already specified by its values along $Z_t$, in the sense that $h+H^\top Z_t=0$ for all $t\ge 0$ if and only if $h=0$ and $H=0$. Similarly, if ${\rm lin}(Z_\cdot)=\R^m$ then any linear function $H^\top Z_t=0$ for all $t\ge 0$ if and only if $H=0$.

\begin{theorem}\label{thmcharLG}
Assume that $Y_t=(\zeta_t;\zeta_t X_t)$ has full linear support, ${\rm lin}(Y_\cdot)=\R^{n+1}$. The following statements are equivalent:
\begin{enumerate}
\item\label{thmcharLG1} $(\zeta_t,X_t)$ forms an LG process;

\item \label{thmcharLG3} $Y_t=(\zeta_t;\zeta_t X_t)$ admits a drift that is strictly linear in $Y_t$,
\begin{equation}\label{eqXMXY}
 dY_t = \kappa Y_t\,dt + dM^Y_t ,
\end{equation}
for some $\kappa\in\R^{(n+1)\times (n+1)}$, and its local martingale part $M^Y_t$ is a martingale. The state price density is given by $\zeta_t = \bm e_1^\top Y_t$, where $\bm e_1^\top =(1,0,\dots,0)$;

\item \label{thmcharLG3new} for any $\alpha\in\R$ and invertible $Q\in\R^{(n+1)\times (n+1)}$, the process $Z_t=\e^{\alpha t}Q(\zeta_t;\zeta_t X_t)$ admits a drift that is strictly linear in $Z_t$,
\begin{equation*}
 dZ_t = \beta Z_t\,dt + dM^Z_t ,
\end{equation*}
for some $\beta\in\R^{(n+1)\times (n+1)}$, and its local martingale part $M^Z_t$ is a martingale. The state price density is given by $\zeta_t = \e^{-\alpha t}\psi^\top Z_t$, where $\psi^\top = \bm e_1^\top Q^{-1}$.
\end{enumerate}

In either case, the functions $\Acal$, $\Bcal$, $\Ccal$, $\Dcal$ in \eqref{eq1} and the matrices $\kappa$ and $\beta$ are related as
\begin{equation}\label{eq8New}
\begin{pmatrix}
\Acal(\tau) & \Bcal(\tau) \\ \Ccal(\tau) & \Dcal(\tau)
\end{pmatrix}=\e^{\kappa\tau}\quad\text{and}\quad \beta = \alpha {\rm Id} + Q\kappa Q^{-1}.
\end{equation}
Moreover, the short rate $r_t$ is a linear function of $X_t$,
\begin{equation}\label{rQeqN}
  r_t = -{A}   - {B} X_t,
\end{equation}
and $X_t$ admits a $\Q$-drift $\mu^{X,\Q}_t$ that is a quadratic function in $X_t$,
\begin{equation}\label{muQeqN}
\mu^{X,\Q}_t =   {C}+ (r_t+{D}) X_t ={C} + ({D}-{A})X_t -({B} X_t) X_t ,
\end{equation}
where $A\in\R$, $B\in\R^{1\times n}$, $C\in\R^{n\times 1}$, and $D\in\R^{n\times n}$ are given in terms of $\kappa$ by
\[\kappa=   \begin{pmatrix}
{A} & {B} \\ {C} & {D}
\end{pmatrix}.\]
\end{theorem}

\begin{remark}
As the proof of Theorem~\ref{thmcharLG} reveals, the implications ${\ref{thmcharLG3new}}\Leftrightarrow{\ref{thmcharLG3}}\Rightarrow{\ref{thmcharLG1}}$ hold without the assumption that ${\rm lin}(Y_\cdot)=\R^{n+1}$.
\end{remark}

Property~{\ref{thmcharLG3}} in Theorem~\ref{thmcharLG} is the definition of an LG process provided in \citeN{Gabaix:2007}, where the quadratic part in the $\Q$-drift of $X_t$ as shown in~\eqref{muQeqN} is referred to as a linearity-generating twist of an AR(1) process.

\begin{remark}
We shall henceforth refer to either pair, $(\zeta_t,X_t)$ or $(Y_t,\zeta_t)$ in Theorem~\ref{thmcharLG}{\ref{thmcharLG3}} or $(Z_t,\zeta_t)$ in Theorem~\ref{thmcharLG}{\ref{thmcharLG3new}}, as $(n+1)$-dimensional LG process.
\end{remark}

While Theorem~\ref{thmcharLG} gives a set of equivalent characterizations of the LG property for the {\emph{given}} pair $(\zeta_t,X_t)$, it does not provide conditions for the {\emph{existence}} of such processes. Indeed, it is nontrivial to specify an LG process $(\zeta_t,X_t)$.

\citeN{CheriditoGabaix2008} and \citeN{CarrGabaixWu2009} specify an $(n+1)$-dimensional LG process $(\zeta_t,X_t)$ by first specifying $Y_t$---that is, the matrix $\kappa$ and the martingale $M^Y_t$---in \eqref{eqXMXY} and then set $\zeta_t = \bm e_1^\top Y_{t}$ and $X_t = Y_{2..n+1,t}/\zeta_t$. More specifically, \citeN{CheriditoGabaix2008} model the factor process $Y_t$ in \eqref{eqXMXY} as a Markov jump-diffusion and provide sufficient conditions on the parameters such that $Y_t$ takes values in the positive orthant $(0,\infty)^{n+1}$. \citeN{CarrGabaixWu2009} specify the martingale part $M^Y_t$ in \eqref{eqXMXY} having an exogenous stochastic volatility component $v_t$, such that $(Y_t,v_t)$ becomes a Markov factor process.

The problem with specifying $Y_t$ as (a component of) the driving factor process is that it does not go well with stationarity. Stationarity of the factor process is desirable in view of statistical model estimation. However, the first component of $Y_t$ is the state price density, which converges to zero in expectation. This suggests that $Y_t$ is not stationary. It is in fact $X_t$ that generally can be thought of as stationary, see also \citeN[Section 3.1]{Gabaix:2007}. This suggests that $X_t$ should be specified as factor process.

However, several issues arise. First, both the state price density and bond prices $P(t,T)$ given by \eqref{eqPtT} have to be positive. Second, if one also requires nonnegative interest rates, which is equivalent to $0<P(t,T)\le 1$ for all $T\ge t$, then the support of $X_t$ has to lie in an intersection of a continuum of half-spaces. Third, $X_t$ has a quadratic $\Q$-drift~\eqref{muQeqN} and a highly nonlinear drift under ${\Pa^\ast}$ in general. Taken together, this makes it difficult to find a priori conditions on the model parameters such that the LG process $(\zeta_t,X_t)$ is well defined and interest rates are nonnegative, or at least bounded from below.

\section{Specification of LR models}\label{secspecLRNew}

We first recall the definition of an LR model introduced by \citeN{FilipovicLarssonTrolle14}.
\begin{definition}
An $m$-dimensional \emph{linear-rational (LR) model} consists of an $m$-dimensional integrable semimartingale $Z_t$ with linear drift and a linear state price density specification,
\begin{equation}\label{mLR}
   dZ_t = (b+\beta Z_t)\,dt + dM^Z_t\quad\text{and}\quad  \zeta_t = \e^{-\alpha t} (\phi+ \psi^\top Z_t),
\end{equation}
for some parameters $b\in\R^m$, $\beta\in\R^{m\times m}$, martingale $M^Z_t$, and parameters $\alpha,\phi \in\R$, $\psi\in\R^m$ such that $\zeta_t>0$ for all $t\ge 0$.
\end{definition}

From Theorem~\ref{thmcharLG}{\ref{thmcharLG3new}} we infer that every $m$-dimensional LG process $(Z_t,\zeta_t)$ can be represented as an $m$-dimensional LR model~\eqref{mLR} with $b=0$ and $\phi=0$. Conversely, it is straightforward to see that every $m$-dimensional LR model~\eqref{mLR} can be represented as an $(m+1)$-dimensional LG process given by
\begin{equation}\label{LRasLGm1new}
Z'_t  =   (1;Z_t) \quad\text{and}\quad  \zeta_t = \e^{-\alpha t}\psi'^\top Z'_t
\end{equation}
with $\psi'  = (\phi;\psi)$. Indeed, it follows by inspection that $Z'_t$ has strictly linear drift,
\[ dZ'_t = \beta' Z'_t \,dt + dM^{Z'}_t,\]
with $(m+1)\times (m+1)$-drift matrix
\[ \beta' = \begin{pmatrix} 0 & 0 \\ b & \beta \end{pmatrix} ,\]
and martingale part $ dM^{Z'}_t = \begin{pmatrix} 0 ;  dM^Z_t\end{pmatrix}$.

\begin{remark}\label{remLGm1}
If $Z_t$ has full linear support, ${\rm lin}(Z_\cdot)=\R^m$, then $Z'_t$ has full linear support, ${\rm lin}(Z'_\cdot)=\R^{m+1}$. But obviously, $Z_t'$ takes values in the affine hyperplane $\{1\}\times\R^m\subset\R^{m+1}$, so that $\dim{\rm aff}(Z'_\cdot)<m+1$.
\end{remark}

In line with~\eqref{eqPtT}, we find that bond prices and short rate become linear-rational in $Z_t$, for $t\le T$,
\begin{equation}\label{PtTZ}
 P(t,T)= \E^\ast_t\left[\frac{\zeta_T}{\zeta_t}\right]= \e^{-\alpha (T-t)} \frac{\phi+\psi^\top\e^{\beta(T-t)}\int_0^{T-t} \e^{-\beta s}b\,ds + \psi^\top\e^{\beta(T-t)} Z_t}{\phi  + \psi^\top Z_t}
\end{equation}
and
\begin{equation}\label{rLR}
  r_t = -\partial_T\log P(t,T)|_{T=t}=\alpha -\frac{\psi^\top(b+\beta Z_t)}{\phi+\psi^\top Z_t}.
\end{equation}
The Radon--Nikodym density process of $\Q$ with respect to ${\Pa^\ast}$ is given by the stochastic exponential $D_t=\Ecal_t(L^D)$ where the local martingale $L^D_t$ is determined by
\begin{equation}\label{NDLR}
  dL^D_t = \frac{\psi^\top dM^Z_t}{\phi+\psi^\top Z_{t-}} .
\end{equation}
Indeed, by definition \eqref{eqspddef} the state price density satisfies $d\zeta_t = -\zeta_t r_t\,dt + \zeta_{t-}dL^D_t$ where $L^D_t$ is the stochastic logarithm of the Radon--Nikodym density process $D_t$. Expanding $\zeta_t$ in \eqref{mLR} and matching drift and martingale terms yields~\eqref{NDLR}.

\begin{remark}\label{remUSV}
It is not required that $Z_t$ has the Markov property. In applications it is often the case that there is an $n$-dimensional semimartingale $U_t$ such that $(Z_t,U_t)$ becomes a Markov process. Such a $U_t$ feeds into the characteristics of the martingale part $M^Z_t$ of $Z_t$. As $U_t$ does not directly appear in the bond price formula~\eqref{PtTZ}, it is \emph{unspanned} by the term structure. The unspanned factor $U_t$ will typically be revealed by prices of bond options. This property of LR models (and LG processes alike) to admit unspanned factors is important in view of the well documented unspanned stochastic volatility phenomenon in financial data, see \citeN{CollinDufresneGoldstein2002}. For more details, we refer the reader to \citeN{FilipovicLarssonTrolle14}.
\end{remark}

A key advantage of LR models is the ease with which they can be specified and be made consistent with nonnegative interest rates. We now sketch how to specify an $m$-dimensional LR model~\eqref{mLR}. Positivity of $\zeta_t$ is achieved by assuming that $\phi>0$ and $\psi^\top Z_t\ge 0$. The latter is tantamount to saying that $Z_t$ takes values in a state space $E$ that is contained in the half-space $\{ z\mid \psi^\top z\ge 0\}$, which is easy to achieve. As dividing the state price density $\zeta_t$ by $\phi$ does not affect prices we may and will henceforth take $\phi = 1$.

 In view of the linear-rational expression \eqref{rLR}, the short rate satisfies $r_t\ge \alpha-\alpha^\ast$ where we denote
\[ \alpha^\ast = \sup_{z\in E} \frac{\psi^\top(b+\beta z)}{1+\psi^\top z}.\]
The value $\alpha^\ast$ is finite under a mild non-degeneracy condition on the model parameters as will be seen in a more specific setup below, see \eqref{condnonnegrates}. Setting $\alpha=\alpha^\ast$ then implies nonnegative interest rates. More generally, we can lower bound interest rates by any level $-\delta$ by setting $\alpha=\alpha^\ast-\delta$. 

As for the finiteness of $\alpha^\ast$, we now assume that the state space of $Z_t$ is the nonnegative orthant, $E=\R^m_+$, and that $\psi\in\R^m_+$. Then there is a simple condition such that $\alpha^\ast$ is finite, and thus interest rates are bounded below. This condition is given in~\citeN[Lemma~5]{FilipovicLarssonTrolle14}, which in the notation of this paper reads as follows. Denote by $\beta_i$ the $i$th column vector of the matrix $\beta$, and let $I$ be the set of indices $i=1,\dots,d$ for which $\psi_i>0$.  We can write
\[ \frac{\psi^\top (b+\beta z)}{\phi+\psi^\top z} = \frac{(\psi^\top b) +\sum_{i=1}^m (\psi^\top\beta_i) z_i}{\phi+\sum_{i\in I}  \psi_i z_i}.\]
From this expression it follows that if
\begin{equation}\label{condnonnegrates}
\psi^\top\beta_i\le 0\quad\text{for all $i\notin I$}
\end{equation}
then $\alpha^\ast$ is finite and given by
\[ \alpha^\ast = \max\left\{ \frac{\psi^\top b}{\phi},\,\frac{\psi^\top\beta_i}{\psi_i},\,i\in I\right\} .\]

There are many ways of specifying an $\R^m_+$-valued semimartingale $Z_t$ with linear drift as in \eqref{mLR}. Examples include any $\R^m_+$-valued component with autonomous drift of an $(m+n)$-dimensional affine jump-diffusion $(Z_t,U_t)$ or polynomial diffusion $(Z_t,U_t)$ studied in \citeN{FilipovicLarsson15}, see also Remark~\ref{remUSV}. This fact together with the simple condition~\eqref{condnonnegrates} provides an easy way of specifying LR models, and hence LG processes, that yield nonnegative interest rates and exhibit unspanned stochastic volatility. Affine jump-diffusion factor processes also have the great advantage that derivatives whose payoffs are nonlinear functions of the state $Z_T$, or more generally $(Z_T,U_T)$, such as interest rate swaptions or more general options on coupon bonds, can be priced efficiently using Fourier transform methods. For more details, we refer the reader to \citeN{FilipovicLarssonTrolle14}.\footnote{The specification of $Y_t$ in \citeN{CarrGabaixWu2009} is also of affine type, albeit time-inhomogeneous, and they use this to derive option pricing formulas via Fourier transforms.}

\section{Relation between LG processes and LR models}\label{secLRNew}

We have seen in \eqref{LRasLGm1new} that every $m$-dimensional LR model can be represented as an $(m+1)$-dimensional LG process. This raises the following questions:

\begin{description}
\item[Q1] Can an LR model also be represented as an LG process of the same dimension?
\item[Q2] Can every LG process be represented as a lower dimensional LR model?
\end{description}

We shall see that the answer is no to both questions in general.

\subsection{Observational equivalence}
We elaborate on questions {\bf Q1} and {\bf Q2} in the context of the following equivalence relation.
\begin{definition}
We say that a $m'$-dimensional LR model
\[    dZ'_t = (b'+\beta' Z'_t)\,dt + dM^{Z'}_t\quad\text{and}\quad
   \zeta'_t = \e^{-\alpha' t} (\phi'+ \psi'^\top Z'_t)\]
is \emph{observationally equivalent} to the $m$-dimensional LR model \eqref{mLR} if the normalized state price densities $\zeta'_t/\zeta'_0=\zeta_t/\zeta_0$ for all $t\ge 0$.
\end{definition}

Observationally equivalent LR models thus have identical normalized state price densities and thus are indistinguishable from an economic point of view.

The following result is straightforward.

\begin{lemma}\label{lemLRphi0}
Any $m$-dimensional LR model~\eqref{mLR} is observationally equivalent to the $m$-dimensional LR model $Z_t'=Z_t-q$ and $\zeta_t'=\zeta_t$ satisfying
\[ dZ_t' = (b' +\beta Z_t')\,dt + dM^Z_t \quad\text{and}\quad \zeta_t'=\e^{-\alpha t}\psi^\top Z_t'\]
with $b'=b+\beta q$, for any $q\in\R^m$ such that $\phi+\psi^\top q=0$.
\end{lemma}

Lemma~\ref{lemLRphi0} implies that we could without of loss of generality assume that $\phi=0$, as long as $\psi\neq 0$.\footnote{The case $\psi=0$ is not of practical interest, as it implies constant interest rates $r_t=\alpha$ and $\Pa^\ast=\Q$, which follows from \eqref{rLR} and \eqref{NDLR}.} The reason why we keep $\phi$ in the representation of the LR model~\eqref{mLR} is that it gives us the flexibility to specify the factor process $Z_t$ on a fixed state space $E$, as shown in Section~\ref{secspecLRNew}. Indeed, the state space is not invariant under the transform $Z'_t= Z_t-q$ in Lemma~\ref{lemLRphi0}.

Here is an immediate consequence of Lemma~\ref{lemLRphi0}.
\begin{corollary}\label{lemLRphi0X1}
An $m$-dimensional LR model~\eqref{mLR} for which there exists some $q\in\R^m$ such that
\begin{equation}\label{LGtypeeq}
\text{$b+\beta q=0$ and $\phi+\psi^\top q=0$}
\end{equation}
is observationally equivalent to the $m$-dimensional LG process $Z_t'=Z_t-q$ and $\zeta_t'=\zeta_t$ satisfying
\[ dZ_t' =  \beta Z_t' \,dt + dM^Z_t \quad\text{and}\quad \zeta_t'=\e^{-\alpha t}\psi^\top Z_t'.\]
\end{corollary}

\subsection{Proper LR models}
Corollary~\ref{lemLRphi0X1} motivates the following definition.
\begin{definition}
An $m$-dimensional LR model~\eqref{mLR} for which there does not exist any $q\in\R^m$ satisfying~\eqref{LGtypeeq} is called \emph{proper}.
\end{definition}
Hence a non-proper LR model is observationally equivalent to an LG process of the same dimension, which partly answers question {\bf Q1}. We will show in Theorem~\ref{thmLRasLGnew} below that also the converse holds. Hereto we first have to rule out some degenerate situations of the following kind. Given an $m$-dimensional LR model~\eqref{mLR} and any $m'>m$, we can trivially generate an observationally equivalent $m'$-dimensional LR model as follows: choose an arbitrary $(m'-m)$-dimensional semimartingale $Z''_t$ with linear drift,
\[ dZ''_t = (b'' + \beta'' Z_t'')\,dt + dM^{Z''}_t, \]
and define $Z'_t=(Z_t; Z''_t)$ and $\zeta_t'=\e^{-\alpha t} (\phi+\psi'^\top Z_t')=\zeta_t$ where $\psi'=(\psi;0)$. Then $Z'_t$ and $\zeta_t'$ obviously form an $m'$-dimensional LR model that is observationally equivalent to \eqref{mLR}. To avoid such redundancies, we want to exclude directions in $\R^m$ that are linearly unspanned by the bond prices~\eqref{PtTZ}. These are directions $\xi\in \R^m$ that when added to $Z_t$ on the right hand side of \eqref{PtTZ} do not affect $P(t,T)$ for any $T\ge t$. We call the space of all such directions the \emph{term structure kernel} and denote it by $\Ucal$. It is shown in \citeN[Theorem~1]{FilipovicLarssonTrolle14} that
\begin{equation}\label{TSKeq}
 \Ucal \supseteq \ker\left\{ \psi^\top \e^{\beta\tau},\,\tau\ge 0\right\}
\end{equation}
with equality if the short rate process $r_t$ given by \eqref{rLR} is not constant.\footnote{\label{footnote1}In view of \eqref{rLR} the short rate process $r_t$ is constant if (and only if) $\psi^\top$ is a left-eigenvector of $\beta$ with eigenvalue $\lambda$ satisfying $\psi^\top b=\lambda\phi$ (assuming that ${\rm aff}(Z_\cdot)=\R^m$). In this case, we have $r_t=\alpha-\lambda$ and the term structure kernel is $\Ucal=\R^m$, while the right hand side of~\eqref{TSKeq} equals $\ker\psi^\top$, a proper subspace of $\Ucal$.}

The following lemma shows that, after a dimension reduction if necessary, we can always and without loss of generality assume that the LR model (LG process) has full affine support (full linear support) and zero term structure kernel.

\begin{lemma} \label{L:ob eq minNEW}
Assume the $m$-dimensional LR model \eqref{mLR} (LG process \eqref{mLR}, with $b=0$ and $\phi=0$) has
\begin{enumerate}
  \item\label{L:ob eq minNEWi} non-full affine support, ${\rm aff}(Z_\cdot)\subsetneq \R^m$ (non-full linear support, ${\rm lin}(Z_\cdot)\subsetneq\R^m$),
\end{enumerate}
or
\begin{enumerate}
\setcounter{enumi}{1}
  \item\label{L:ob eq minNEWii} non-zero term structure kernel, $\Ucal\neq\{0\}$, and non-constant short rate process.
\end{enumerate}
Then there exists an observationally equivalent $m'$-dimensional LR model (LG process) $Z_t'$ and $\zeta_t'=\zeta_t$ with $m'<m$, full affine support, ${\rm aff}(Z'_\cdot)= \R^{m'}$ (full linear support, ${\rm lin}(Z'_\cdot)=\R^{m'}$) and zero term structure kernel, $\Ucal'=\{0\}$.
\end{lemma}

The following theorem provides a full characterization of proper LR models. It answers in the negative question {\bf Q1}, as proper LR models obviously exist.

\begin{theorem}\label{thmLRasLGnew}
Assume that the $m$-dimensional LR model \eqref{mLR} has zero term structure kernel, $\Ucal=\{0\}$, and full affine support, ${\rm aff}(Z_\cdot)=\R^m$. Then there does not exist any observationally equivalent lower dimensional LG process. Moreover, the following are equivalent:
\begin{enumerate}
\item\label{thmLRasLGnew1} the LR model \eqref{mLR} is proper;
\item \label{thmLRasLGnew2} there does not exist any observationally equivalent $m$-dimensional LG process;
\item\label{thmLRasLGnew3} the observationally equivalent $(m+1)$-dimensional LG process \eqref{LRasLGm1new} has zero term structure kernel.
\end{enumerate}
In either case, there also does not exist any observationally equivalent lower dimensional LR model.
\end{theorem}

\begin{remark}
As the proof of Theorem~\ref{thmLRasLGnew} reveals, the equivalence ${\ref{thmLRasLGnew1}}\Leftrightarrow{\ref{thmLRasLGnew3}}$ holds without the assumption that ${\rm aff}(Z_\cdot)=\R^m$.
\end{remark}

Theorem~\ref{thmLRasLGnew} stipulates that there exist non-proper LR models with zero term structure kernel and full affine support that are observationally equivalent to some lower dimensional LR models, which shows that the converse of Lemma~\ref{L:ob eq minNEW} does not hold. Indeed, we can construct such examples as follows. Modifying \eqref{LRasLGm1new}, for any $m$-dimensional LR model \eqref{mLR} we define the observationally equivalent $(m+1)$-dimensional LG process
\begin{equation}\label{LRasLGm1mod}
Z'_t  =  \e^{(\alpha-\alpha')t}  \begin{pmatrix} 1 ;   Z_t\end{pmatrix} \quad\text{and}\quad
 \zeta'_t = \e^{-\alpha' t}\psi'^\top Z'_t
\end{equation}
with $\psi'  = (\phi;\psi)$. It follows by inspection that $\zeta'_t=\zeta_t$ and that $Z'_t$ has strictly linear drift,
\[ dZ'_t = \beta' Z'_t \,dt + dM^{Z'}_t,\]
with $(m+1)\times (m+1)$-drift matrix
\[ \beta' = (\alpha-\alpha'){\rm Id}+\begin{pmatrix} 0 & 0 \\ b &  \beta \end{pmatrix} ,\]
and martingale part $ dM^{Z'}_t = \e^{(\alpha-\alpha')t} \begin{pmatrix} 0 ;  dM^Z_t\end{pmatrix}$. If $Z_t$ has full linear support ${\rm lin}(Z_\cdot)=\R^m$ then $Z'_t$ has full linear support ${\rm lin}(Z'_\cdot)=\R^{m+1}$. If $\alpha'=\alpha$, we are back to \eqref{LRasLGm1new} and $Z'_t$ takes values in the affine hyperplane $\{1\}\times\R^m\subset\R^{m+1}$, so that $\dim{\rm aff}(Z'_\cdot)<m+1$, see Remark~\ref{remLGm1}. But if $\alpha'\neq\alpha$ then $Z'_t$ can have full affine support ${\rm aff}(Z'_\cdot)=\R^{m+1}$ and zero term structure kernel, as the following example shows.

\begin{example}\label{exmfaff}
Let $m=1$ and consider the $1$-dimensional LR model given by the square-root diffusion $dZ_t = b\, dt + \sqrt{Z_t}\,dW_t $ and $\zeta_t = \e^{-\alpha t} (\phi +\psi Z_t)$, with $Z_0>0$ and parameters $b,\alpha,\phi,\psi>0$. Let $\alpha'>\alpha$, then the observationally equivalent $2$-dimensional LG process \eqref{LRasLGm1mod} has full affine support, ${\rm aff}(Z'_\cdot)=\R^2$, and zero term structure kernel, $\Ucal'=\{0\}$.
\end{example}

\subsection{Reducible LG processes}

Example \ref{exmfaff} shows an $m$-dimensional LG process (non-proper LR model) with zero term structure kernel and full affine support that is observationally equivalent to an $(m-1)$-dimensional LR model. We now formalize the class of all LG processes for which there exist observationally equivalent LR models of lower dimension.
\begin{definition}
An $m$-dimensional LG process for which there exists an observationally equivalent $m'$-dimensional LR model with $m'<m$ is called \emph{reducible}.
\end{definition}

The following theorem provides a full characterization of reducible LG processes. It also shows that every reducible LG process is of the form~\eqref{LRasLGm1new}, up to observational equivalence.

\begin{theorem}\label{thmrLGLRrepr}
Assume that the $m$-dimensional LG process \eqref{mLR}, with $b=0$ and $\phi=0$,  has zero term structure kernel, $\Ucal=\{0\}$, and full linear support, ${\rm lin}(Z_\cdot)=\R^m$. Then the following are equivalent:

\begin{enumerate}
  \item\label{thmLGLRrepr1} the LG process \eqref{mLR} is reducible;
  \item\label{thmLGLRrepr2} there exists some $v\in\R^m$ such that $v^\top Z_t = v^\top Z_0\e^{\lambda t}$ for some real constant $\lambda$ and $v^\top Z_0>0$;
  \item\label{thmLGLRrepr3} there exists an observationally equivalent $(m-1)$-dimensional LR model $Z_t'$ and $\zeta_t'$ such that the $m$-dimensional LG process given by $(1;Z_t')$ and $\zeta_t'$, as in \eqref{LRasLGm1new}, is observationally equivalent to the LG process \eqref{mLR}.
\end{enumerate}
In either case, any observationally equivalent $m'$-dimensional LR model $Z_t'$ and $\zeta_t'$ with $m'<m$ satisfies $m'=m-1$, is proper, has full affine support, ${\rm aff}(Z'_\cdot)=\R^{m-1}$, and zero term structure kernel, $\Ucal'=\{0\}$, and $v$ is a left eigenvector of $\beta$ with eigenvalue~$\lambda$.
\end{theorem}

The following example shows a non-reducible LG process, which answers in the negative question {\bf Q2}. It also shows that the mapping~\eqref{LRasLGm1new} is not surjective.

\begin{example}\label{exnrLG}
Consider the 2-dimensional LG process given by
\begin{align*}
  dZ_{1t}&= Z_{2t}\,dt\\
  dZ_{2t}&= c Z_{2t}\,dt + \sqrt{Z_{2t}}\,dW_t
\end{align*}
for some constant $c>0$ and a standard Brownian motion $W_t$, with initial value $Z_0\in (0,\infty)^2$, and $\zeta_t = \e^{-\alpha t} \bm 1^\top Z_t$. Accordingly, the drift matrix is
\[ \beta=\begin{pmatrix}
  0 & 1 \\ 0 & c
\end{pmatrix}.\]
It follows by inspection that the term structure kernel is zero, $\Ucal=\{0\}$. We claim that the LG process is not reducible. Indeed, for any nonzero $v\in\R^2$ we have that
\[ v^\top dZ_t = (v_1+c v_2) Z_{2,t}\,dt + v_2\sqrt{Z_{2,t}}\,dW_t, \]
which is not of the form $v^\top Z_0 \,d(\e^{\lambda t}) $ for any real $\lambda$. Whence $Z_t$ has full linear support, ${\rm lin}(Z_\cdot)=\R^2$, and Theorem~\ref{thmrLGLRrepr} yields the claim.
\end{example}

In Section~\ref{secLTF}, we show that reducible LG processes and proper LR models appear naturally in the context of long-term risk factorization in the spirit of \citeN{AlvarezJermann2005}.

\section{Long-term risk factorization}\label{secLTF}

The state price density specification in a proper LR model goes naturally with the long-term risk factorization due to \citeN{AlvarezJermann2005} and \citeN{HansenScheinkman2009} in Markovian environments, and extended to a general semimartingale environment in \citeN{qin_lin_17}.

The \emph{$T$-forward measure} $\Q^T\sim\Pa^\ast$ related to the state price density $\zeta_t$ is defined by the Radon--Nikodym density process $\frac{d\Q^T}{d\Pa^\ast}|_{\Fcal_t}=M(t,T)$ given by
\begin{equation}\label{MtTdef}
M(t,T)= \frac{\zeta_t P(t,T)}{\E^\ast[\zeta_0 P(0,T)]}=\frac{\E^\ast_t[\zeta_T]}{\E^\ast[\zeta_T]}\quad\text{for all $t\ge 0$,}
\end{equation}
where we set $P(t,T)=\E^\ast_t[\zeta_T]/\zeta_t$ for all $t\ge 0$, which amounts to investing the notional of the $T$-bond in the savings account at $t=T$, so that $M(t,T)=M(T,T)$ for $t>T$.\footnote{Equation~\eqref{MtTdef} corresponds to \citeN[Equation (2.2)]{qin_lin_17} for constant $\zeta_0$. Indeed, a standing assumption in \citeN[Sections 2--3]{qin_lin_17} is that $\Fcal_0$ is trivial under $\Pa^\ast$.} Under $\Q^T$ any asset price process $S_t$ discounted by the $T$-bond, $S_t/P(t,T)$, becomes a martingale, because the deflated price process $\zeta_t S_t$ is a $\Pa^\ast$-martingale.

\citeN[Theorem 3.1]{qin_lin_17} show that if the limit
\begin{equation}\label{assLFM}
\text{$M^\infty_t =\lim_{T\to\infty} M(t,T)$ exists in $L^1$ and $M^\infty_t>0$ for all $t\ge 0$,}
\end{equation}
then $M^\infty_t$ is a positive martingale with $\E^\ast[M^\infty_t]=1$ and defines the \emph{long forward measure} $\La$ which is equivalent to $\Pa^\ast$ on each $\Fcal_t$ with Radon--Nikodym density $M^\infty_t$.\footnote{The long forward measure $\La$ is only locally defined on each $\Fcal_t$ but not globally on $\Fcal_\infty$, unless $M^\infty_t$ is uniformly integrable.}\footnote{Even if $\La$ exists on $\Fcal_\infty$, the identity between instantaneous $T$-forward rates and conditional expectation of future short rates $r_T$ under the $T$-forward measure, $ f(t,T)=\E^T_t[r_T]=\E^\ast_t[M(T,T) r_T]$, does not carry over to infinite maturity $T\to\infty$. That is, $f(t,\infty)= \lim_{T\to\infty}\E^T_t[r_T]  \neq \lim_{T\to\infty}\E^\La_t[r_T]$, in general. An example is given by any proper one-dimensional LR model $dZ_t=(b+\beta Z_t)\,dt + \sigma\sqrt{Z_t}\,dW_t$ with $\beta<0$ that is ergodic with unconditional mean $\theta$ such that $\Pa^\ast=\La$ is the long-forward measure. If $Z_0$ has the stationary distribution of $Z_t$, then $\lim_{T\to\infty}\E^\ast_t[r_T] = \E^\ast[r_0]$. It then follows that $f(t,\infty)=\alpha$ and $\E^\ast[r_0]<\alpha$ by the convexity of the linear-rational function on the right hand side of~\eqref{rLR}.}
As a consequence, the \emph{long bond} exists,
\[ B^\infty_t = \lim_{T\to\infty}\frac{P(t,T)}{P(0,T)} = \frac{\zeta_0}{\zeta_t} \lim_{T\to\infty} \frac{M(t,T)}{M(0,T)}= \frac{\zeta_0}{\zeta_t} \frac{M^\infty_t}{M^\infty_0},\]
with limit in probability for all $t\ge 0$. $B^\infty_t$ can be interpreted as the gross return earned by time $t$ on the investment of one dollar at time zero in the zero-coupon bond of asymptotically long maturity.  Under $\La$ any asset price process $S_t$ has a smaller conditional expected log return than $B^\infty_t$. Indeed, Jensen's inequality shows
\[ \E^{\La}_t\left[\log \left(\frac{S_T}{S_t}\right)\right]-\E^{\La}_t\left[\log \left(\frac{B^\infty_T}{B^\infty_t}\right)\right] = \E^{\La}_t\left[\log \left(\frac{ S_T  B^\infty_t}{S_t  B^\infty_T}\right)\right]  \le\log \E^{\La}_t\left[\frac{ \zeta_T S_T}{ \zeta_t S_t} \frac{ M^\infty_t}{M^\infty_T}\right] =0 .\]
Hence $B^\infty_t$ is the $\La$-growth optimal portfolio.\footnote{\cite{det_rin_10} discuss the role of the long bond in the optimal portfolios of long-horizon investors, dealing in equities and bonds, with von Neumann--Morgenstern preferences.} We conclude that the state price density admits the long-term factorization
\[  \zeta_t  =  \frac{\zeta_0}{ B^\infty_t M^\infty_0}  M^\infty_t  .\]
The first factor $\zeta_0/(B^\infty_t M^\infty_0 )$ is referred to as \emph{transitory component} and the martingale $M^\infty_t$ as \emph{permanent component} in the terminology of \citeN{AlvarezJermann2005}. The former is the implied state price density with respect to the long forward measure $\La$. Identity $\zeta_t=\zeta_0/(B^\infty_t M^\infty_0 )$ holds if and only if $M^\infty_t=1$ for all $t\ge 0$, or equivalently if the auxiliary measure $\Pa^\ast=\La$ on each $\Fcal_t$. As the above derivation was under the assumption~\eqref{assLFM}, we are led to the following result.

\begin{lemma}\label{lemPilfm}
The auxiliary measure $\Pa^\ast$ is the long forward measure related to the state price density $\zeta_t$ if and only if $M(t,T)\to 1$ in $L^1$ as $T\to\infty$ for all $t\ge 0$. In this case the long bond is given by $B^\infty_t=\zeta_0/\zeta_t$.
\end{lemma}

This has practical implications for building models. Suppose we have specified the state price density $\zeta_t$ such that the auxiliary measure $\Pa^\ast$ is the long forward measure. The state price density $\zeta^\Pa_t $ with respect to the objective measure $\Pa$ is then obtained though the long-term factorization $\zeta^\Pa_t=\zeta_t \hat M_t$ by an exogenous specification of the permanent component $\hat M_t$. This is the approach in \citeN{FilipovicLarssonTrolle14}.\footnote{\citeN{qin_lin_18} estimate the permanent component in the long-term factorization using U.S.\ Treasury data. They find that the empirically estimated permanent component is highly volatile and cannot be omitted. This is consistent with the empirical findings of \citeN{FilipovicLarssonTrolle14} based on swap data that incorporating the permanent component is critical for generating realistic risk premium dynamics.}

We now study the long-term risk factorization in LR models. The \emph{long-term yield} is defined by
\[
y_\infty(t) = - \lim_{T\to\infty} \frac{1}{T-t}\log P(t,T),
\]
with limit in probability whenever it exists.

\begin{lemma}\label{lemLRnr}
Assume that the $m$-dimensional LR model \eqref{mLR} is proper and that all eigenvalues of the drift matrix $\beta$ have negative real part. Then $\Pa^\ast$ is the long forward measure and the long-term yield is constant equal to $y_\infty(t) = \alpha$.
\end{lemma}

The following trivial example shows that the converse of Lemma~\ref{lemLRnr} does not hold.

\begin{example}
Consider any deterministic $1$-dimensional LR model given by $dZ_t=(b+\beta Z_t)\,dt$ and $\zeta_t=\e^{-\alpha t}(\phi+\psi Z_t)$. Then trivially $\Pa^\ast$ is the long forward measure, irrespective of the sign of $\beta$ or whether the LR model is proper. Moreover, the long-term yield is constant equal to $y_\infty(t) = \alpha-\beta^+$. In particular, for $\beta=0$ we obtain $y_\infty(t) = \alpha$, hence the converse of Lemma~\ref{lemLRnr} does not hold.
\end{example}

The following example shows a proper LR model for which $\Pa^\ast$ is the long forward measure but the long-term yield does not exist. It shows that the assumption in Lemma~\ref{lemLRnr} that the eigenvalues of $\beta$ have negative real part cannot be relaxed for asserting the existence of the long-term yield. It further shows that the assumption in \citeN[Theorem 3.2]{qin_lin_17} cannot be relaxed.

\begin{example}
Consider the deterministic 3-dimensional LR model given by
\[
dZ_t = \begin{pmatrix}0 &\varphi& 0\\-\varphi &0 &0\\0& 0& -\kappa\end{pmatrix}Z_t\,dt, \qquad Z_0 = \begin{pmatrix} 1\\1\\1 \end{pmatrix},
\]
for some $\varphi>0$ and $\kappa>0$, so that
\[
Z_t=\begin{pmatrix} \cos(\varphi t)-\sin(\varphi t) \\ \cos(\varphi t)+\sin(\varphi t) \\ \e^{-\kappa t} \end{pmatrix}.
\]
Let $\alpha\in\R$, $\phi=1$, and $\psi=(1/2;1/2;1)$. The state price density is then
\[
\zeta_t = \e^{-\alpha t}(1 + \cos(\varphi t) + \e^{-\kappa t}),
\]
which is positive for all $t\ge0$. Inspection shows that this LR model is proper. Because it is deterministic, $\Pa^\ast$ is the long forward measure. On the other hand, we have $P(0,T)=\zeta_T/\zeta_0$ and hence
\begin{equation} \label{eq:no y_infty}
-\frac{1}{T}\log P(0,T) = \alpha - \frac{1}{T}\log(1 + \cos(\varphi T) + \e^{-\kappa T}) + \frac{1}{T}\log 3.
\end{equation}
Letting $T$ tend to infinity along the sequence $T_n=2n\pi /\varphi$, the right-hand side of~\eqref{eq:no y_infty} converges to $\alpha$. Choosing instead the sequence $T_n=(2n+1)\pi/\varphi$ we obtain the limit $\alpha+\kappa$. Thus the long-term yield at time zero, $y_\infty(0)$, does not exist.
\end{example}

Our main result of this section shows that LG processes under the long forward measure are always reducible, under some slightly stronger linear support condition.

\begin{theorem}\label{thmGREAT}
Assume that the $m$-dimensional LG process \eqref{mLR}, with $b=0$ and $\phi=0$, has zero term structure kernel, $\Ucal=\{0\}$, that ${\rm lin}(\e^{-\beta \cdot}Z_\cdot)=\R^m$, and that $\Pa^*$ is the long forward measure. Then the following hold:
\begin{enumerate}
\item\label{thmGREAT1} the LG process \eqref{mLR} is reducible,
\item\label{thmGREAT2} the long-term yield is constant equal to $y_\infty(t) = \alpha-\lambda$, where $\lambda$ is the real eigenvalue of $\beta$ given in Theorem~\ref{thmrLGLRrepr}, and $\lambda$ is larger than or equal to the real parts of all eigenvalues of $\beta$,
\item\label{thmGREAT3} $Z_t$ has full linear support, ${\rm lin}(Z_\cdot)=\R^m$.
\end{enumerate}
\end{theorem}

The following example shows a non-reducible LG process for which $\Pa^\ast$ is not the long forward measure. In fact, it does not admit a long forward measure at all and its long-term yield exists but is not constant.

\begin{example}
Consider the non-reducible 2-dimensional LG process $Z_t$ and $\zeta_t$ given in Example~\ref{exnrLG}. The process $Z_t$ takes values in $(0,\infty)\times [0,\infty)$. Denote by $\tau_0=\inf\{ t\mid Z_{2t}=0\}$ the hitting time of zero of the square-root diffusion $Z_{2t}$. It satisfies $\Pa^\ast[\tau_0<\infty]>0$ and zero is an absorbing point, $Z_{2t}=0$ for all $t\ge \tau_0$, see, e.g., \citeN[Exercise~10.12]{fil_09}. An elementary calculation shows
\[ \e^{\beta t} = \begin{pmatrix}
  1 & \frac{\e^{c t}-1}{c} \\ 0 & \e^{ct}
\end{pmatrix} ={\rm Id} + \begin{pmatrix}
  0 & \frac{\e^{c t}-1}{c} \\ 0 & \e^{ct}-1
\end{pmatrix},\]
for all real $t$, so that
\[ \E^\ast_t[\zeta_T] = \e^{-\alpha T}\left(\bm 1^\top Z_t + Z_{2t} \left(\e^{c(T-t)}-1\right)\left( 1/c +1\right)\right) \]
and ${\rm lin}(\e^{-\beta \cdot}Z_\cdot)=\R^2$. Hence the linear support condition of Theorem~\ref{thmGREAT} is met. We obtain that the long-term yield exists but is not constant and non-decreasing,
\begin{align*}
   y_\infty(t) &= -\lim_{T\to\infty} \frac{1}{T-t} \log\left( \E^\ast_t[\zeta_T]/\zeta_t\right) \\
   & = \alpha - \lim_{T\to\infty} \frac{1}{T-t} \left(\log\left( \bm 1^\top Z_t + Z_{2t} \left(\e^{c(T-t)}-1\right)\left( 1/c +1\right) \right) -\log(\zeta_t)\right)\\
   &=\begin{cases} \alpha-c,& \text{if $t<\tau_0$} \\ \alpha,&\text{if $t\ge \tau_0$}\end{cases}
\end{align*}
with limits in probability. This is consistent with the Dybvig--Ingersoll--Ross theorem that asserts that the long-term yield can never fall under the absence of arbitrage, see \citeN{dyb_ing_ros_96}. Similarly, we see that the long forward measure does not exist. Indeed,
\begin{align*}
  M(t,T) &= \frac{\bm 1^\top Z_t + Z_{2,t} \left(\e^{c(T-t)}-1\right)\left( 1/c +1\right)}{\bm 1^\top \E^\ast[Z_0] + \E^\ast[Z_{2,0}] \left(\e^{cT}-1\right)\left( 1/c +1\right)}\to \frac{\e^{-ct} Z_{2,t}}{\E^\ast[Z_{2,0}]}=M^\infty_t
\end{align*}
in $L^1$ as $T\to\infty$. The martingale $M^\infty_t$ becomes zero for $t\ge\tau_0$, so that \eqref{assLFM} does not hold and $M^\infty_t$ does not define an equivalent measure on $\Fcal_t$ for any $t\ge 0$.

\end{example}

\section{Conclusion}\label{secconc}

We have reviewed LG processes and LR models. We have shown that every LR model can be represented as an LG process, and vice versa, subject to a dimensionality adjustment which we have studied in detail. This has useful practical implications. The direct specification of an LG process is arguably a difficult task for the modeler---even more so if interest rates were to be bounded below. LR models can be easily specified and allow for direct control of the lower bound on interest rates. Moreover, LR models appear naturally in the context of the long-term risk factorization due to \citeN{AlvarezJermann2005}, \citeN{HansenScheinkman2009}, and \citeN{qin_lin_17}. The range of flexible specifications of LR models inducing bounded below interest rates is wide and will be the subject of future research in empirical asset pricing.

\begin{appendix}

\section{Proof of Theorem~\ref{thmcharLG}}

${\ref{thmcharLG1}}\Rightarrow{\ref{thmcharLG3}}$: it is readily seen that $(\zeta_t,X_t)$ forms an LG process if and only if
\begin{equation}\label{eqXXnew}
\E^\ast_t[Y_T]=\Mcal(T-t) Y_t
\end{equation}
for all $0\le t\le T<\infty$, for the continuously differentiable matrix-valued function
\[ \Mcal(\tau)=\begin{pmatrix}
\Acal(\tau) & \Bcal(\tau) \\ \Ccal(\tau) & \Dcal(\tau)
\end{pmatrix}.\]
Taking nested conditional expectations, this implies
\[ \Mcal(\tau_1+\tau_2) Y_t = \E^\ast_t[Y_{t+\tau_1+\tau_2}]=\E^\ast_t[\E^\ast_{t+\tau_2}[Y_{t+\tau_1+\tau_2}]]=\Mcal(\tau_1)\Mcal(\tau_2)Y_t \]
for all $\tau_1,\tau_2,t\ge 0$. As ${\rm lin}(Y_\cdot)=\R^{n+1}$, we conclude that $\Mcal(\tau_1+\tau_2)=\Mcal(\tau_1)\Mcal(\tau_2)$ for all $\tau_1,\tau_2\ge 0$ and hence the first equality in \eqref{eq8New} for some $\kappa\in\R^{(n+1)\times (n+1)}$. It follows that $N_t=\e^{-\kappa t}Y_t$ is a martingale. Applying It\^o's formula to $Y_t=\e^{\kappa t}N_t$ yields \eqref{eqXMXY}, so that $Y_t$ has a drift, and where $dM^Y_t = \e^{\kappa t}\,dN_t$ is a martingale. Indeed, integration by parts gives
\[ M^Y_T - M^Y_t = \e^{\kappa T} N_T - \e^{\kappa t} N_t - \int_t^T \e^{\kappa s} \kappa N_s \,ds\]
for $t \le T$. Taking $\Fcal_t$-conditional expectation and changing the order of integration, justified by the fact that $\E^\ast[ \| \e^{\kappa s} \kappa N_s \| ]\le \e^{\|\kappa\|T}\|\kappa\|\E^\ast[ \| N_T \| ] <\infty$ for all $s\le T$, we obtain
\[ \E^\ast_t[ M^Y_T - M^Y_t ] = \left( \e^{\kappa T} - \e^{\kappa t} - \int_t^T \e^{\kappa s} \kappa \,ds \right) N_t = 0,\]
as desired.

Moreover, after an equivalent change of measure we see that $N'_t = N_t/D_t = \e^{-\kappa t -\int_0^t r_s \,ds} X'_t$ is a $\Q$-martingale, where we denote $X'_t=(1;X_t)$. Applying It\^o's formula to $X'_t=\e^{\kappa t +\int_0^t r_s \,ds}N'_t$ shows that $X'_t$, and thus $X_t$, has a $\Q$-drift that is of the form
\[ \mu^{X',\Q}_t = (0;\mu^{X,\Q}_t) = (\kappa + r_t) X'_t . \]
Matching terms implies \eqref{rQeqN} and \eqref{muQeqN}.

${\ref{thmcharLG3}}\Rightarrow{\ref{thmcharLG1}}$: If $Y_t$ satisfies \eqref{eqXMXY} with a martingale $M^Y_t$ then it readily follows that
\[ \E^\ast_t [ Y_T ] = \e^{(T-t)\kappa}Y_t ,\]
that is \eqref{eqXXnew}, which implies {\ref{thmcharLG1}}.

${\ref{thmcharLG3}}\Leftrightarrow{\ref{thmcharLG3new}}$: This follows from the relation $Z_t = \e^{\alpha t} Q Y_t$, which also implies the second equality in \eqref{eq8New}. The proof of Theorem~\ref{thmcharLG} is thus complete.

\section{Proof of Lemma~\ref{L:ob eq minNEW}}

We prove the LR case of the lemma, and consider an $m$-dimensional LR model~\eqref{mLR}. We prove the two statements
\begin{enumerate}
\item\label{L:ob eq minNEW:proof1} if ${\rm aff}(Z_\cdot)\subsetneq \R^m$, then there exists an observationally equivalent $m'$-dimensional LR model $Z_t'$ and $\zeta_t'=\zeta_t$ with $m'<m$ and ${\rm aff}(Z'_\cdot)= \R^{m'}$;
\item\label{L:ob eq minNEW:proof2} if $\Ucal\neq\{0\}$ and the short rate process is non-constant, then there exists an observationally equivalent $m'$-dimensional LR model $Z_t'$ and $\zeta_t'=\zeta_t$ with $m'<m$ and $\Ucal'=\{0\}$.
\end{enumerate}

We first prove \ref{L:ob eq minNEW:proof1}, and assume that ${\rm aff}(Z_\cdot)\subsetneq \R^m$. Define $A={\rm aff}(Z_\cdot)$ and $m'=\dim A<m$. Let $z\mapsto P(z-q)$ with $P\in\R^{m'\times m}$ and $q\in A$ be an invertible affine map from $A$ to $\R^{m'}$. Denote the inverse map by $z'\mapsto Qz'+q$ with $Q\in\R^{m\times m'}$. Note that $Z_t\in A$ for all $t\ge0$. We now specify the new factor process $Z'_t=P(Z_t-q)$, which has the linear drift dynamics
\[
dZ'_t = (P (b+\beta q) + P\beta Q Z'_t)\,dt + P \,dM^Z_t.
\]
Furthermore, the state price density $\zeta_t$ can be written
\[
\zeta_t = \e^{-\alpha t}(\phi + \psi^\top q + \psi^\top QP Z'_t).
\]
By specifying $\zeta'_t=\zeta_t$ we thus obtain an observationally equivalent $m'$-dimensional LR model. It is clear that ${\rm aff}(Z'_\cdot)=\R^{m'}$, as otherwise $Q({\rm aff}(Z'_\cdot))+q$ would be a proper affine subspace of $A={\rm aff}(Z_\cdot)$ which contains $Z_t$ for all $t\ge0$, a contradiction. The observationally equivalent $m'$-dimensional LR model $Z_t'$ and $\zeta_t'$ thus satisfies $m'<m$ and ${\rm aff}(Z'_\cdot)=\R^{m'}$, which proves \ref{L:ob eq minNEW:proof1}.

We next prove \ref{L:ob eq minNEW:proof2}, and assume that $\Ucal\neq\{0\}$ and that the short rate process is non-constant. We then have equality in \eqref{TSKeq}, so that
\[
\Ucal=\ker\left\{ \psi^\top \e^{\beta\tau},\,\tau\ge 0\right\}\ne\{0\}
\]
Let $n'=\dim\Ucal>0$ and define $m'=m-n'<m$. Choose an invertible matrix $P=(P_1;\, P_2)\in\R^{(m'+n')\times m}$ such that $\Ucal=\ker P_1$ and $P_2(\Ucal)=\R^{n'}$. Denote the inverse matrix by $Q=(Q_1,\, Q_2)\in\R^{m\times(m'+n')}$, so that in particular we have
\begin{align}
P_1Q_1&={\rm id}, \label{eq:LR min 0} \\
Q_1P_1 + Q_2P_2 &= {\rm id}. \label{eq:LR min 0.1}
\end{align}
Moreover, as $Q_2$ maps $\R^{n'}$ to $\Ucal$, as $\beta$ leaves $\Ucal$ invariant, and as $\Ucal$ lies in $\ker P_1$ and in $\ker \psi^\top$, we have
\begin{equation}\label{eq:LR min 1}
\psi^\top Q_2=0 \qquad\text{and}\qquad P_1\beta^k Q_2 = 0, \qquad k\ge0.
\end{equation}
We now specify the factor process $(Z'_t; U_t)=PZ_t$, with $Z'_t=P_1Z_t$ and $U_t=P_2Z_t$. In view of~\eqref{eq:LR min 1} we obtain the linear drift dynamics
\[
dZ'_t = (P_1 b + P_1 \beta Q_1 Z'_t)\,dt + P_1 \,dM^Z_t,
\]
and the state price density $\zeta_t$ can be written
\[
\zeta_t =\e^{-\alpha t}(\phi + \psi^\top Q_1Z'_t).
\]
By specifying $\zeta'_t=\zeta_t$ we thus obtain an observationally equivalent $m'$-dimensional LR model.

We claim that the term structure kernel~$\Ucal'$ of this model is zero. We first claim that
\begin{equation} \label{eq:LR min 2}
(P_1\beta Q_1)^k = P_1\beta^k Q_1
\end{equation}
for all $k\ge1$. For $k=1$ this is obvious. Suppose~\eqref{eq:LR min 2} holds for some given power $k\ge1$. Then, in view of~\eqref{eq:LR min 0.1} and~\eqref{eq:LR min 1}, we get
\[
(P_1\beta Q_1)^{k+1} = P_1\beta^k Q_1P_1\beta Q_1 = P_1\beta^{k+1}Q_1 - P_1\beta^k Q_2P_2\beta Q_1 = P_1\beta^{k+1}Q_1.
\]
Thus~\eqref{eq:LR min 2} holds with $k$ replaced by $k+1$, and by induction for all $k\ge1$.

Now let $\xi'\in\Ucal'$ be arbitrary. Note that the observationally equivalent $m'$-dimensional LR model $Z_t'$ and $\zeta_t'$ has drift matrix $\beta'=P_1\beta Q_1$ and state price density vector $\psi'=Q_1^\top \psi$, and, due to the non-constant short rate,
\[
\Ucal' = \ker\left\{ \psi'^\top \e^{\beta'\tau},\,\tau\ge 0\right\}.
\]
Therefore, using~\eqref{eq:LR min 0} and~\eqref{eq:LR min 2}, we obtain
\[
\psi^\top \e^{\beta\tau}Q_1\xi' = \psi^\top Q_1P_1\e^{\beta\tau}Q_1\xi' = \psi'^\top \e^{\beta'\tau}\xi' = 0, \qquad \tau\ge0.
\]
Thus $Q_1\xi'\in\Ucal$. On the other hand, $\Ucal\subseteq\ker P_1$, whence $\xi'=P_1Q_1\xi'=0$. Hence $\Ucal'=\{0\}$ is zero, as claimed. The observationally equivalent $m'$-dimensional LR model $Z_t'$ and $\zeta_t'$ thus satisfies $m'<m$ and $\Ucal'=\{0\}$, which proves \ref{L:ob eq minNEW:proof2}.

The LR case of the lemma now follows from \ref{L:ob eq minNEW:proof1} and \ref{L:ob eq minNEW:proof2}. Indeed, suppose for contradiction that no observationally equivalent LR model with full affine support and zero term structure kernel exists. Then \ref{L:ob eq minNEW:proof1} or \ref{L:ob eq minNEW:proof2} can be applied infinitely many times, each time resulting in an observationally equivalent LR model of strictly lower dimension. This is a contradiction, because the dimension $m$ of the original LR model is finite.

The LG case of the lemma, $b=0$ and $\phi=0$, is proved similarly. In \ref{L:ob eq minNEW:proof1} we simply replace affine supports by linear supports and note that one can take $q=0$. In \ref{L:ob eq minNEW:proof2} no changes are needed. In both cases we then observe that the constructed $m'$-dimensional LR models are in fact LG processes.

\section{Proof of Theorem~\ref{thmLRasLGnew}}

We first prove a lemma which shows that, under the assumption of a zero term structure kernel, observationally equivalent LR models are also algebraically related. The following proofs build on this result.

\begin{lemma}\label{thmnmLReqm'LR}
Assume that the $m$-dimensional LR model \eqref{mLR} has zero term structure kernel, $\Ucal=\{0\}$. Then  any observationally equivalent $m'$-dimensional LR model
\begin{equation}\label{m'LRnew}
   dZ'_t = (b'+\beta' Z'_t)\,dt + dM^{Z'}_t\quad\text{and}\quad
   \zeta'_t = \e^{-\alpha' t} (\phi'+ \psi'^\top Z'_t)
\end{equation}
satisfies
\begin{equation}\label{mLReqm'LR}
Z_t +p  = \frac{\zeta_0}{\zeta'_0}\e^{(\alpha-\alpha') t} (q+Q Z'_t)
\end{equation}
for some vectors $p,q\in\R^{m}$ and some $m\times m'$-matrix $Q$. Moreover, the following properties hold:
\begin{enumerate}
\item\label{thmmLReqm'LR1} If \eqref{mLR} is an LG process, that is, if $b=0$ and $\phi=0$, then $p=0$;
\item\label{thmmLReqm'LR2}  If \eqref{m'LRnew} is an LG process, that is, if $b'=0$ and $\phi'=0$, then $q=0$;
\item\label{thmmLReqm'LR3}  If \eqref{mLR} has full affine support, ${\rm aff}(Z_\cdot)=\R^m$, then $m'\ge m-1$;
\item\label{thmmLReqm'LR4}  If \eqref{mLR} has full affine support, ${\rm aff}(Z_\cdot)=\R^m$, and \eqref{m'LRnew} is an LG process then $m'\ge m$ and $\dim{\rm lin}(Z'_\cdot)\ge m$;
\item\label{thmmLReqm'LR5}  If \eqref{mLR} is an LG process and has full linear support, ${\rm lin}(Z_\cdot)=\R^m$, and \eqref{m'LRnew} is an LG process then $m'\ge m$ and $\dim{\rm lin}(Z'_\cdot)\ge m$.
\end{enumerate}
\end{lemma}

\begin{remark}
Example \ref{exmfaff} illustrates Lemma~{\ref{thmnmLReqm'LR}\ref{thmmLReqm'LR3}}: an $m$-dimensional LR model with zero term structure kernel and full affine support, which is observationally equivalent to an $(m-1)$-dimensional LR model.
\end{remark}

\begin{proof}
As $\Ucal=\{0\}$ by assumption, it follows from \eqref{TSKeq} that we have
\begin{equation} \label{eq:TSK_000}
\ker\left\{ \psi^\top \e^{\beta\tau},\,\tau\ge 0\right\}=\{0\}.
\end{equation}
Consider an observationally equivalent $m'$-dimensional LR model \eqref{m'LRnew}. Matching the $\Fcal_t$-conditional expectations of $\zeta_{t+\tau}/\zeta_0$ and $\zeta'_{t+\tau}/\zeta'_0$ gives
\begin{multline}\label{eqXX}
   \frac{1}{\zeta_0}\left(\e^{-\alpha (t+\tau)} \left( \phi + \psi^\top \int_0^\tau \e^{\beta s} b\,ds \right) + \e^{-\alpha (t+\tau)} \psi^\top \e^{\beta\tau}Z_t\right)\\
   =\frac{1}{\zeta'_0}\left(\e^{-\alpha' (t+\tau)} \left( \phi' + \psi'^\top \int_0^\tau \e^{\beta' s} b'\,ds \right) + \e^{-\alpha' (t+\tau)} \psi'^\top \e^{\beta'\tau}Z'_t\right)
 \end{multline}
for all $t,\tau\ge 0$. Due to \eqref{eq:TSK_000}, there exist $0\le \tau_1<\cdots<\tau_m$ such that the row vectors $\e^{-\alpha \tau_1}\psi^\top \e^{\beta\tau_1},\dots,\e^{-\alpha \tau_m}\psi^\top \e^{\beta\tau_m}$ are linearly independent. If we stack the corresponding $m$ equations \eqref{eqXX} in matrix form and invert we obtain
\[ \frac{1}{\zeta_0}\e^{-\alpha t} (Z_t +p)  = \frac{1}{\zeta'_0}\e^{-\alpha' t} (q+Q Z'_t) \]
for some vectors $p,q\in\R^{m}$ and some $m\times m'$-matrix $Q$, which is equivalent to \eqref{mLReqm'LR} .

If \eqref{mLR} is an LG process, so that $\phi=0$ and $b=0$, we infer that $p=0$, which proves property~{\ref{thmmLReqm'LR1}}. Similarly, if \eqref{m'LRnew} is an LG process, so that $\phi'=0$ and $b'=0$, we infer that $q=0$, which proves property~{\ref{thmmLReqm'LR2}}. Properties {\ref{thmmLReqm'LR3}}--{\ref{thmmLReqm'LR5}} now follow by inspection. For \ref{thmmLReqm'LR4} and \ref{thmmLReqm'LR5} we use that ${\rm lin}(Z'_\cdot)={\rm lin}(V_\cdot Z'_\cdot)$ for any positive scalar process $V_t>0$.
\end{proof}

We now prove Theorem~\ref{thmLRasLGnew}. Lemma~\ref{thmnmLReqm'LR}\ref{thmmLReqm'LR4} shows that there does not exist any observationally equivalent lower dimensional LG process. Furthermore, the implication $\ref{thmLRasLGnew2} \Rightarrow \ref{thmLRasLGnew1}$ follows directly from Corollary~\ref{lemLRphi0X1} by considering the contrapositive statement. To prove $\ref{thmLRasLGnew1}\Rightarrow\ref{thmLRasLGnew2}$ we again consider the contrapositive, and assume there exists an observationally equivalent $m'$-dimensional LG process \eqref{m'LRnew}, with $b'=0$ and $\phi'=0$, such that $m'=m$. There is no loss of generality to assume that $\zeta_0'=\zeta_0$ and $\alpha'=\alpha$. Indeed, we may otherwise consider the new factor process and state price density
\[
Z''_t=\frac{\zeta_0}{\zeta'_0}\e^{(\alpha-\alpha') t} Z'_t \quad\text{and}\quad \zeta''_t=\frac{\zeta_0}{\zeta'_0}\zeta'_t = \e^{-\alpha t}\psi'^\top Z''_t
\]
to obtain an observationally equivalent $m$-dimensional LG process with the desired properties. We thus assume that $\zeta_0'=\zeta_0$ and $\alpha'=\alpha$. Lemma~\ref{thmnmLReqm'LR}\ref{thmmLReqm'LR2} then yields
\[
Z_t + p = QZ'_t
\]
for some $p\in\R^m$ and some $m\times m$-matrix $Q$. As $Z_t$ has full affine support, $Q$ is invertible, so by equating the drifts of $Z_t+p$ and $QZ'_t$ we obtain
\[
b+\beta Z_t = Q\beta' Q^{-1}p + Q\beta' Q^{-1} Z_t.
\]
Using that $Z_t$ has full affine support yields $\beta=Q\beta' Q^{-1}$ and then $b=Q\beta' Q^{-1}p=\beta p$. Next, as the two models are observationally equivalent with $\zeta'_0=\zeta_0$ and $\alpha'=\alpha$, we get $\phi + \psi^\top Z_t = \psi'^\top Z'_t$ and hence
\[
\phi + \psi^\top Z_t = \psi'^\top Q^{-1} p + \psi'^\top Q^{-1} Z_t.
\]
Using that $Z_t$ has full affine support yields $\psi^\top=\psi'^\top Q^{-1}$ and then $\phi=\psi'^\top Q^{-1} p=\psi^\top p$. We have thus proved that $b-\beta p=0$ and $\phi-\psi^\top p=0$, showing that the LR model \eqref{mLR} is not proper. This proves $\ref{thmLRasLGnew1}\Rightarrow\ref{thmLRasLGnew2}$.

It remains to prove $\ref{thmLRasLGnew1}\Leftrightarrow\ref{thmLRasLGnew3}$. To this end, first observe that, as $\Ucal=\{0\}$ and ${\rm aff}(Z_\cdot)=\R^m$, the short rate process is non-constant, see Footnote~\ref{footnote1}. Therefore we have equality in \eqref{TSKeq}, and the term structure kernel $\Ucal'$ of the observationally equivalent $(m+1)$-dimensional LG process \eqref{LRasLGm1new} satisfies
\[
\Ucal' = \ker\left\{ \psi'^\top \e^{\beta'\tau},\,\tau\ge 0\right\}.
\]
Consequently, a vector $\xi'=(\delta;\xi)\in\R^{1+m}$ lies in $\Ucal'$ if and only if
\begin{equation} \label{eq L obs eq 2}
0 =  \psi'^\top \e^{\beta'\tau}\xi' = \left(\phi + \psi^\top\int_0^\tau \e^{\beta s}b\,ds\right)\delta +  \psi^\top \e^{\beta\tau}\xi, \quad \tau\ge0,
\end{equation}
where we used the elementary fact that
\[
\e^{\beta'\tau}=\begin{pmatrix} 1 & 0 \\ \int_0^\tau \e^{\beta s} b \,ds & \e^{\beta\tau}\end{pmatrix}.
\]
For $\delta=0$, \eqref{eq L obs eq 2} is equivalent to $\xi\in\Ucal$, and thus equivalent to $\xi=0$. It follows that an element of $\Ucal'$ is nonzero if and only if it is a scalar multiple of some vector $(1;\xi)$ satisfying
\[
\phi + \psi^\top\int_0^\tau \e^{\beta s}b\,ds +  \psi^\top \e^{\beta\tau}\xi = 0, \quad \tau\ge0.
\]
By rewriting the left-hand side, we infer that $\Ucal'$ is nonzero if and only if there exists some $\xi\in\R^m$ such that
\begin{equation} \label{thmLRasLGnew1_eq1}
\phi+\psi^\top \xi + \psi^\top\int_0^\tau \e^{\beta s}(b+\beta\xi)\,ds = 0, \quad \tau\ge0.
\end{equation}
If this is the case we obtain $\phi+\psi^\top\xi=0$ by setting $\tau=0$, and $b+\beta\xi=0$ by differentiating with respect to $\tau$ and setting $\tau=0$. Thus the $m$-dimensional LR model \eqref{mLR} is not proper. Conversely, if the $m$-dimensional LR model \eqref{mLR} is not proper, so that \eqref{LGtypeeq} holds for some $q\in\R^m$, it follows that \eqref{thmLRasLGnew1_eq1} holds for $\xi=q$, whence $\Ucal'$ is nonzero. This proves $\ref{thmLRasLGnew1}\Leftrightarrow\ref{thmLRasLGnew3}$.

It remains to prove the last statement in Theorem~\ref{thmLRasLGnew}. Suppose, by contradiction, that there exists an observationally equivalent $m'$-dimensional LR model \eqref{m'LRnew} with $m'<m$. By Lemma~\ref{thmnmLReqm'LR}{\ref{thmmLReqm'LR3}} we have $m'=m-1$. But then \eqref{LRasLGm1mod}, with $Z_t'$ in lieu of $Z_t$ on the right hand side, defines an observationally equivalent $m$-dimensional LG process, which contradicts the assumption. This completes the proof of the theorem.

\section{Proof of Theorem~\ref{thmrLGLRrepr}}
{\ref{thmLGLRrepr1}}$\Rightarrow${\ref{thmLGLRrepr2}}: Consider an observationally equivalent $m'$-dimensional LR model~\eqref{m'LRnew} with $m'<m$. Lemma~\ref{thmnmLReqm'LR}\ref{thmmLReqm'LR1} implies that
\begin{equation}\label{eqXXT}
 Z_t    = \frac{\zeta_0}{\zeta'_0}\e^{(\alpha-\alpha') t} (q+Q Z'_t)
\end{equation}
for some $ q\in\R^{m}$ and some $m\times m'$-matrix $Q$. As ${\rm lin}(Z_\cdot)=\R^m$ we infer that $m'=m-1$ and there exists a $v\in\R^m$ such that $v^\top q>0 $ and $v^\top Q=0$. Property~{\ref{thmLGLRrepr2}} follows by left-multiplying~\eqref{eqXXT} by $v^\top$. In view of Lemma~\ref{thmnmLReqm'LR}\ref{thmmLReqm'LR5} there does not exist any observationally equivalent $(m-1)$-dimensional LG process, hence the LR model~\eqref{m'LRnew} is proper. It also follows that the LR model~\eqref{m'LRnew} has full affine support, ${\rm aff}(Z'_\cdot)=\R^{m-1}$, and zero term-structure kernel, $\Ucal'=\{0\}$. Because otherwise, by Lemma~\ref{L:ob eq minNEW}, one could find an observationally equivalent $(m-2)$-dimensional LR model, which again would induce through \eqref{LRasLGm1new} an observationally equivalent $(m-1)$-dimensional LG process, contradicting that the LR model~\eqref{m'LRnew} is proper.

{\ref{thmLGLRrepr2}}$\Rightarrow${\ref{thmLGLRrepr3}}: Taking conditional expectation, we obtain for any $t\le T$
\[ v^\top \e^{\beta (T-t)} Z_t= \E^\ast_t[ v^\top Z_T]=v^\top Z_0\e^{\lambda T}= \e^{\lambda(T-t)} v^\top Z_t.\]
As ${\rm lin}(Z_\cdot)=\R^m$, we conclude that $v$ is a left eigenvector of $\beta$ with eigenvalue $\lambda$. This proves the last statement in the theorem.

Now let $Q$ be an invertible $m\times m$-matrix whose first row is $v^\top$, and define $\widetilde Z_t = \e^{-\lambda t} Q Z_t/(v^\top Z_0)$. Then $\widetilde Z_t$ has a strictly linear drift,
\[
d\widetilde Z_t = Q(\beta-\lambda)Q^{-1}\widetilde Z_t\, dt + \e^{-\lambda t}/(v^\top Z_0) Q \,dM^Z_t,
\]
and the first component of $\widetilde Z_t$ is constant and equal to one. Let $Z'_t$ consist of the last $m-1$ components of $\widetilde Z_t$, so that $\widetilde Z_t=(1;Z_t')$. Set $\alpha'=\alpha-\lambda$ and $\psi'=(Q^\top)^{-1}\psi$, then $\zeta'_t=\e^{-\alpha' t} \psi'^\top (1;Z_t') = \zeta_t/(v^\top Z_0)$ is seen to define an observationally equivalent $(m-1)$-dimensional LR model of the desired form.

{\ref{thmLGLRrepr3}}$\Rightarrow${\ref{thmLGLRrepr1}}: This holds by definition, which completes the proof of Theorem~\ref{thmrLGLRrepr}.

\section{Proof of Lemma~\ref{lemLRnr}}

As $\beta$ is invertible, we have
\[
\E^\ast_t[\phi+\psi^\top Z_T] = \phi-\psi^\top\beta^{-1}b + \psi^\top \e^{\beta (T-t)} (Z_t+\beta^{-1}b),
\]
which converges to $\phi-\psi^\top\beta^{-1}b$ in $L^1$ as $T\to\infty$ because all eigenvalues of $\beta$ have negative real part. As the LR model is proper, we have $\phi-\psi^\top\beta^{-1}b > 0$ and therefore
\[
M(t,T)= \frac{\E^\ast_t[\phi+\psi^\top Z_T]}{\E^\ast[\phi+\psi^\top Z_T]} \to 1\quad\text{ in $L^1$ as $T\to\infty$.}
\]
By Lemma~\ref{lemPilfm}, this proves that $\Pa^*$ is the long forward measure. Moreover, we have
\[
- \frac{1}{T-t}\log P(t,T) = \alpha - \frac{1}{T-t}\log \E^\ast_t[\phi+\psi^\top Z_T] + \frac{1}{T-t}\log\left(\phi+\psi^\top Z_t\right) \to \alpha
\]
in probability as $T\to\infty$. Thus $y_\infty(t)=\alpha$ as claimed.

\section{Proof of Theorem~\ref{thmGREAT} }

The proof of Theorem~\ref{thmGREAT} builds on the following lemma.

\begin{lemma} \label{L:thmGREAT}
Assume that the $m$-dimensional LG process \eqref{mLR}, with $b=0$ and $\phi=0$, has zero term structure kernel, $\Ucal=\{0\}$, full linear support, ${\rm lin}(Z_\cdot)=\R^m$, and constant long-term yield $y_\infty(t)=\alpha$. Then the eigenvalues of $\beta$ have nonpositive real parts.
\end{lemma}

\begin{proof}
The equality
\[
- \frac{1}{T-t}\log P(t,T) = \alpha - \frac{1}{T-t}\log \psi^\top \e^{\beta(T-t)}Z_t + \frac{1}{T-t}\log\psi^\top Z_t
\]
along with the assumption that $y_\infty(t)=\alpha$ yields
\[
\lim_{\tau\to\infty} \frac{1}{\tau}\log \psi^\top \e^{\beta\tau}Z_t  = 0
\]
in probability, and hence almost surely, because $Z_t$ does not depend on $\tau$. Consequently, for any $\varepsilon>0$ we have
\[
 \psi^\top \e^{\beta \tau} Z_t \le \e^{\varepsilon\tau}
\]
for all $\tau\ge\tau_0$, where $\tau_0$ depends on $t$, $\omega$, and $\varepsilon$. As $Z_t$ has full linear support we infer that
\begin{equation} \label{psibound}
\| \psi^\top \e^{\beta\tau} \| \le c_1 \e^{\varepsilon\tau} \quad\text{for all $\tau\ge0$}
\end{equation}
for some positive constant $c_1$. Stacking the trivial identities $\psi^\top\e^{\beta\tau_i}\e^{\beta\tau}=\psi^\top\e^{\beta(\tau_i+\tau)}$ for suitable values of $\tau_1,\ldots,\tau_m$ and using that the term structure kernel is zero so that equality holds in \eqref{TSKeq}, we get
\[
\e^{\beta\tau} = A^{-1} B(\tau),\quad\text{for $A=\begin{pmatrix}\psi^\top\e^{\beta\tau_1} \\ \vdots \\ \psi^\top\e^{\beta\tau_m} \end{pmatrix}$ and $B(\tau)=  \begin{pmatrix}\psi^\top\e^{\beta(\tau_1+\tau)} \\ \vdots \\ \psi^\top\e^{\beta(\tau_m+\tau)} \end{pmatrix}$.}
\]
In view of \eqref{psibound}, the operator norm of $\e^{\beta\tau}$ is therefore bounded by
\[
\|\e^{\beta\tau}\| \le  \|A^{-1}  \|\,  \| B(\tau) \|
\le  \|A^{-1} \|\, \max_{i=1,\ldots,m} \| \psi^\top \e^{\beta(\tau_1+\tau)} \|
\le c_2 \e^{\varepsilon\tau}
\]
for all $\tau\ge0$ and some positive constant $c_2$. On the other hand, every eigenvalue $\lambda$ of $\beta$ satisfies $\e^{{\rm Re\,} \lambda \tau}\le\|\e^{\beta\tau}\|$, and hence ${\rm Re\,}\lambda\le\varepsilon$. As $\varepsilon>0$ was arbitrary it follows that ${\rm Re\,}\lambda\le0$ as claimed.
\end{proof}

We can now prove Theorem~\ref{thmGREAT}.
We first prove \ref{thmGREAT3}. Indeed, let $v\in\R^m$ be such that $v^\top Z_T=0$ for all $T\ge 0$. Taking conditional expectation we obtain $v^\top\e^{\beta T} (\e^{-\beta t}Z_t)=0$ for all $t\le T$. By assumption there exists some $T\ge 0$ such that ${\rm lin}(\e^{-\beta t}Z_t ,t\le T)=\R^m$. We conclude that $v^\top \e^{\beta T}=0$ and hence $v=0$, so that $Z_t$ has full linear support, which proves~\ref{thmGREAT3}.

Define the row vector valued function
\[
F(T) = \frac{\psi^\top \e^{\beta T}}{\psi^\top \e^{\beta T} \E^\ast[Z_0]}.
\]
In view of Lemma~\ref{lemPilfm}, and because $\Pa^\ast$ is the long forward measure, we have
\begin{equation} \label{lim is one}
\text{$\lim_{T\to\infty} F(T) \e^{-\beta t}Z_t = \lim_{T\to\infty} \frac{\psi^\top \e^{\beta (T-t)}Z_t}{\psi^\top \e^{\beta T} \E^\ast[Z_0]} = 1$ in $L^1$ for all $t\ge0$.}
\end{equation}
As ${\rm lin}(\e^{-\beta \cdot}Z_\cdot)=\R^m$, this implies that
\begin{equation} \label{F(T)conv}
\lim_{T\to\infty}F(T)=v^\top
\end{equation}
for some nonzero vector $v\in\R^m$. Consequently,
\begin{equation} \label{F'(T)conv}
F'(T) = F(T)\beta - F(T)\beta \E^\ast[Z_0] F(T)  \to v^\top \beta - v^\top \beta \E^\ast[Z_0] v^\top
\end{equation}
as $T\to\infty$.

For any real-valued $C^1$ function $f(T)$ on $[0,\infty)$ such that $a=\lim_{T\to\infty}f(T)$ and $a'=\lim_{T\to\infty}f'(T)$ both exist and are finite, one necessarily has $a'=0$. Indeed, $a'>0$ would imply that $f'(T)\ge a'/2>0$ for all $T$ greater than some finite $T_0$, which gives the contradiction
\[
f(T)=f(T_0) + \int_{T_0}^T f'(t)\,dt \ge f(T_0) + \frac{a'}{2}(T-T_0) \to \infty
\]
as $T\to\infty$. This shows that $a'\le0$, and one similarly finds $a'\ge0$. In view of \eqref{F(T)conv} and \eqref{F'(T)conv} we may apply this to the components of $F(T)$ to get $\lim_{T\to\infty}F'(T)=0$ and then
\[
v^\top \beta = \lambda v^\top \quad \text{where} \quad \lambda = v^\top \beta \E^\ast[Z_0].
\]
This shows that $v$ is a left eigenvector of $\beta$ with real eigenvalue $\lambda$. Plugging this back in \eqref{lim is one}, combined with \eqref{F(T)conv}, gives
\[ 1= v^\top \e^{-\beta t}Z_t = \e^{-\lambda t} v^\top Z_t\]
and Theorem~\ref{thmrLGLRrepr} shows that the LG process \eqref{mLR}, with $b=0$ and $\phi=0$, is reducible. This proves~\ref{thmGREAT1}.

Next, to compute the long-term yield we write
\[
\log P(t,T) = - \alpha(T-t) + \log \psi^\top \e^{\beta T} \E^\ast[Z_0] + \log \frac{\psi^\top \e^{\beta(T-t)} Z_t}{\psi^\top \e^{\beta T}\E^\ast[Z_0]} - \log\psi^\top Z_t.
\]
As the ratio in the second to last term converges to $1$ in $L^1$ due to \eqref{lim is one}, we get
\[
-\frac{1}{T-t}\log P(t,T) = \alpha - \frac{\log \psi^\top \e^{\beta T} \E^\ast[Z_0]}{T-t} + o(1).
\]
where $o(1)\to 0$ in probability as $T\to\infty$. Furthermore, \eqref{F(T)conv} gives
\[
\frac{d}{dT} \log \psi^\top \e^{\beta T} \E^\ast[Z_0] = F(T)\beta \E^\ast[Z_0] \to v^\top \beta \E^\ast[Z_0] = \lambda
\]
as $T\to\infty$. L'H\^opital's rule thus yields
\[
y_\infty(t) =  \alpha - \lim_{T\to\infty} \frac{\log \psi^\top \e^{\beta T} \E^\ast[Z_0]}{T-t} = \alpha-\lambda
\]
in probability, as claimed.\footnote{We use the following form of L'H\^opital's rule: Let $f(T)$ and $g(T)$ be differentiable with with $g'(T)\ne0$ on $(0,\infty)$. If $a=\lim_{T\to\infty}f'(T)/g'(T)$ exists and $ |g(T)|\to \infty$ as $T\to\infty$, then $\lim_{T\to\infty}f(T)/g(T)=a$. }

It remains to prove that $\lambda$ is larger than or equal to the real parts of all eigenvalues of $\beta$. To this end, consider the observationally equivalent LG process given by $Z'_t=\e^{-\lambda t}Z_t$ and $\zeta_t'=\e^{-\alpha' t}\psi^\top Z_t'=\zeta_t$ with $\alpha'=\alpha-\lambda$. This model has zero term structure kernel, full linear support, and constant long-term yield $y_\infty(t)=\alpha'$. Lemma~\ref{L:thmGREAT} thus implies that the eigenvalues of $\beta'=\beta-\lambda$ have nonpositive real parts. Thus $\lambda$ is larger than or equal to the real parts of all eigenvalues of $\beta$, as required.  This proves \ref{thmGREAT2} and completes the proof of Theorem~\ref{thmGREAT}.

\end{appendix}

\bibliographystyle{chicago}
\bibliography{abt2_DF}

\end{document}